\documentclass[conference]{IEEEtran}

\IEEEoverridecommandlockouts
\usepackage{cite}
\usepackage{amsmath,amssymb,amsfonts}
\usepackage{algorithmic}
\usepackage{graphicx}
\usepackage{textcomp}
\usepackage{xcolor}
\usepackage{epstopdf}
\usepackage{multirow}
\usepackage{url}
\usepackage{amsthm}

\usepackage{tabularx}
\usepackage{float}
\usepackage{subfig}
\usepackage{times, epsfig}
\usepackage[colorinlistoftodos]{todonotes}

\newcommand{\ES}[1]{\ensuremath{\mathsf{E}\left[#1 \right]}}

\newcommand{\paren}[1]{\left(#1\right)}
\newcommand{\sqparen}[1]{\left[#1\right]}
\newcommand{\brparen}[1]{\left\{#1\right\}}

\newcommand{\e}[1]{\ensuremath{{\rm e}^{#1}}} 

\DeclareMathOperator*{\argmax}{argmax}
\renewcommand{\vec}[1]{\ensuremath{\boldsymbol{#1}}} 
\newcommand{\ie}{\ensuremath{{\text{\em i.e.}}}}
\newcommand{\eg}{\ensuremath{{\text{\em e.g.}}}}

\newtheorem{lemma}{Lemma}

\renewcommand{\vec}[1]{\ensuremath{\boldsymbol{#1}}} 

\newlength{\figwidth}
\begin{document}

\title{
Physical-Layer Security for Untrusted UAV-Assisted Full-Duplex Wireless Networks 
}

\author{
\IEEEauthorblockN{Theshani Nuradha\IEEEauthorrefmark{1}, Kasun T. Hemachandra\IEEEauthorrefmark{1}, Tharaka Samarasinghe\IEEEauthorrefmark{1}\IEEEauthorrefmark{2} and Saman Atapattu\IEEEauthorrefmark{2}\\}
\IEEEauthorblockA{
\IEEEauthorrefmark{1}Department of Electronic and Telecommunication Engineering, University of Moratuwa, Moratuwa, Sri Lanka \\
\IEEEauthorrefmark{2}Department of Electrical and Electronic Engineering, University of Melbourne, Victoria, Australia \\
\IEEEauthorblockA{
Email: \IEEEauthorrefmark{1}theshani.nuradha@gmail.com; \IEEEauthorrefmark{1}kasunh@uom.lk; \IEEEauthorrefmark{1}tharakas@uom.lk; \IEEEauthorrefmark{2}saman.atapattu@unimelb.edu.au
}
}
}
\maketitle
\begin{abstract}
 
The paper considers physical layer security (PLS) of an untrusted unmanned aerial vehicle (UAV) network, where a multitude of UAVs communicate in full-duplex (FD) mode.  A  source-based jamming (SBJ) scheme is exploited for secure communication without utilizing any external jammers. Firstly, the optimal power allocation between the confidential signal and the jamming signal is derived to maximize the secrecy rate of each link between the source and the destination. Then, the best UAV selection scheme is proposed to maximize the overall secrecy rate of the network. The corresponding secrecy outage probability (SOP) and the average secrecy rate (ASR) of the network are analyzed based on the proposed UAV selection and the optimal power allocation schemes. Asymptotic results are also obtained to derive the achievable diversity order. Results are validated through numerical evaluations while providing useful network design insights such as locations and altitudes of the UAVs. 
\end{abstract}
\section{Introduction}
Several new technologies have been developed during the past decade to to fulfill the different requirements of 5G and beyond wireless networks, e.g., massive multiple-input multiple-output (MIMO), non-orthogonal multiple access (NOMA), millimeter wave, cognitive and cooperative communications, energy harvesting, backscatter communications, Unmanned aerial vehicles (UAVs) communications, to name a few \cite{Asadpour2014magcom,zhao2018coml,atapattu2017coml}.
Among them, UAVs have been identified as a promising technique for wireless communication purposes due to their deployment flexibility. While UAVs have a broad range of applications including  military, surveillance and monitoring and delivery, UAV-assisted wireless connectivity for telecommunications has been considered as a key enabler for rescue operations,
specially when the existing infrastructure is unavailable. In emergency networks for disaster management, UAVs can act as aerial mobile relays to facilitate information exchange between affected areas and remote data centers or base stations \cite{Zhao2019wcom,MultipleUAVrelays}. 
However, UAV networks are also vulnerable to potential attacks from eavesdroppers. Thus, physical-layer security (PLS), which exploits the physical characteristics of the wireless environment, is crucial for reliable communications. 

\subsection{Related work}
For trusted UAVs, PLS has been considered in the recent research efforts \cite{Wang2017coml,Wang2018acc,Zhong2019coml,Liu2018jcn,Zhou2018tvt,Li2019coml,Tang2019ifs}. 
In \cite{Wang2017coml}, a UAV-enabled relay network under an external eavesdropper was investigated to maximize the average achievable secrecy rate by optimizing the transmit power allocation among the flight period. 
In \cite{Zhong2019coml}, a cooperative jamming scheme was introduced to secure the UAV communication by leveraging on jamming from other nearby UAVs. For multiple UAV relays, an opportunistic relaying in the presence of multiple UAV eavesdroppers was investigated in \cite{Liu2018jcn}, where the UAV-transmitter and UAV-relay pair with the highest end-to-end signal-to-noise ratio (SNR) is selected. 
Moreover, UAVs can also participate as friendly jammers for unsecured ground communication links, $\eg$, \cite{Zhou2018tvt,Li2019coml}, or as external eavesdroppers for ground communications, $\eg$, \cite{Tang2019ifs}.  
The secrecy performance of an UAV-to-UAV system was also studied under a group of UAV eavesdroppers in \cite{Ye2019wcoml}.  

\subsection{Problem statement}
It is important to note that all these research efforts are on {\it trusted half-duplex UAVs}. 
Due to the sudden and/or unplanned deployments, similar to traditional relays \cite{Yener2010it,Atapattu2019twc}, the performance of UAV networks may be disrupted by untrusted behavior of UAVs. Therefore, in this paper, we study an untrusted  UAV network. Since, with a proper self-interference cancellation, full-duplex (FD) communications achieve a higher multiplexing gain than half-duplex (HD) communications, $\eg$, \cite{Wang2018jsac}, FD communication is one of the key transmission techniques for 5G and
beyond applications. We thus consider FD UAVs as well.     
To the best of our knowledge, untrusted UAV-assisted HD or FD communications have not been investigated on the PLS perspective in the open literature, and is still an open problem. To fill in this research gap, this paper studies the PLS of an {\it untrusted FD multiple UAV} network. 
Apart from the untrusted UAVs, we assume that other malicious users such as external jammers do not affect the performance of the network. Therefore, unlike in previous work on external jamming, we consider source based jamming (SBJ) \cite{Atapattu2019coml}. 
The receiving node which may be a remote base station can also consist of multiple antennas. 

\subsection{Contributions}
Based on some fundamental results in \cite{Atapattu2019coml}, we first deduce the optimal power allocation at the ground transmitter with respect to the secrecy rate. We then consider the best UAV selection scheme in order to maximize the secrecy rate. The best single-relay selection is important because it reduces synchronization issues at the receiver while achieving a higher spatial diversity \cite{Atapattu13_jsac}. Based on single-UAV selection and the optimal power allocation schemes, we derive the secrecy outage probability (SOP) in closed-form, and also analyze the average secrecy rate (ASR) over Nakagami-$m$ multi-path fading and standard UAV path-loss model. Moreover, we provide insightful asymptotic results to further investigate the PLS performance of an UAV-assisted network.



\section{System Model}\label{s_sys}
\begin{figure}[h]
	\centering
    \includegraphics[width=0.5\figwidth]{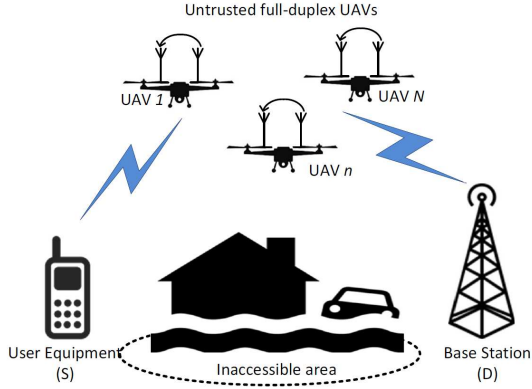}
	\caption{A wireless network with single untrusted FD UAV.}
	\label{f_sysmod}
\end{figure}
We consider a UAV-assisted wireless communication network. A set of $R$ untrusted UAVs ($\mathcal{U}$), whose locations are fixed, is deployed as aerial relays to support uplink communication for user equipment (UE). We consider an arbitrary UE, hereafter referred to as the source node ($\mathcal{S}$), whose uplink communication is blocked by obstacles as shown in Fig.~\ref{f_sysmod}. For example, consider a scenario where terrestrial infrastructure is destroyed due to a natural disaster. $\mathcal{S}$ may still access the core network through a far away base station, hereafter referred to as the destination ($\mathcal{D}$), with the aid of the UAVs. Since such an ad-hoc network is not pre-planned, the supportive UAVs may not be trustworthy. Therefore, we assume that the UAVs assist communication as amplify-and-forward (AF) relays. To improve the throughput, the relays operate in the in-band full-duplex (FD) mode. A single transmit antenna is used by $\mathcal{S}$, while $\mathcal{D}$ is equipped with a maximal-ratio combiner (MRC) with $N$ receive antennas. The UAVs are equipped with an antenna for transmission and an antenna for reception.

With SBJ, $\mathcal{S}$ transmits a composite signal containing the confidential signal $x_s$ and a jamming signal $x_j$. $P_s$ is the power budget of $\mathcal{S}$ and $a\in [0,1]$ is the power allocation ratio. Thus, the power in $x_s$ and $x_j$ is given by $a P_s$ and $(1-a)P_s$, respectively. Furthermore, both signals are of unit average energy. We assume that $\mathcal{D}$ has full knowledge of $x_j$, and full channel state information (CSI), {\em i.e.}, for each $i \in \mathcal{U}$, the CSI of the link between $\mathcal{S}$ and $i$, and the CSI of the link between $i$ and $\mathcal{D}$. This enables jamming signal cancellation and MRC based detection. We assume that the UAVs do not hinder the CSI acquisition process \cite{Xiong2016twc}.


For UAV $i \in \mathcal{U}$, let $S_i$ denote the link between $\mathcal{S}$ and $i$, and let $D_i$ denote the link between $i$ and $\mathcal{D}$. We use $|\cdot|$ to represent the spatial distance between the transmitter and the receiver of such a link. The fading channel amplitudes of these links are modeled as independent and identically distributed Nakagami-$m$ random variables (RVs) \cite{NakagamiM}. Similarly, for UAV $i \in \mathcal{U}$, $H_i$ denotes the altitude, and $d_{S,i}$  and $d_{D,i}$ denote the distances from $\mathcal{S}$ and $\mathcal{D}$ to the projection of $i$ on the ground plane, respectively. 
Following \cite{ChannelModelUAV}, the average  path loss for a ground to UAV (or UAV to ground) link of length $r$ is given by
 $\bar{l}_{r}=r^{\alpha}\paren{4\pi f/c}^{2}\paren{\eta_{\rm L}p_{\rm L}+\eta_{\rm N}p_{\rm N}}$,
where $
  p_{L}=1/\sqparen{1+\omega\textnormal{ exp}\paren{-\beta\sqparen{\theta-\omega}}},
 p_{\rm N}=1-p_{\rm L},
  \theta =\frac{180}{\pi}\arctan\paren{\frac{H_i}{d_i}}, r=\sqrt{H_i^{2}+d_i^{2}}; d_i\in\{d_{S,i},d_{D,i}\}; $
$f$ is the carrier frequency, $c=3\times 10^8$ m/s, $\alpha$ is the path loss exponent, $\eta_{\rm L}$, $\eta_{\rm N}$, $\omega$, and $\beta$ are constants that depend on factors such as blockage density, height and the density of surrounding buildings. 

At time $\tau$, the received signal at UAV $i \in \mathcal{U}$ is given by
\begin{equation}\label{eq:rx UAV}
\begin{split}
    y_{i}(\tau)=\sqparen{\sqrt{\frac{a P_s}{\bar{l}_{|S_i|}}}x_s(\tau)+\sqrt{\frac{(1-a) P_s}{\bar{l}_{|S_i|}}}x_j(\tau)}h_{S_i} \\ +I(\tau)+n_{i}(\tau),
\end{split}
\end{equation}
where $h_{S_i}$ denote the fading channel coefficient of link $S_i$, $n_{i}(\tau)$ is the additive white Gaussian noise (AWGN) with variance $\sigma_u^2$, and $I(\tau)$ is the residual self-interference (RSI). The RSI is assumed to be independent of other signals, and follows a Gaussian distribution with zero-mean and variance $\sigma_I^2$ \cite{Ngoc2014jsac}. 

Assuming variable gain AF relaying, the amplification gain of $i \in \mathcal{U}$ is given by 
$G_i =\sqrt{P_u}\sqparen{P_s|h_{S_i}|^2/\bar{l}_{|S_i|} + \sigma_I^2 + \sigma_{u}^2}^{-1/2}$,
where $P_{u}$ is the transmit power of the UAV. 
Thus, the transmit signal from $i$ is $G_i y_{i}(\tau-\Delta)$, where $\Delta$ is the processing delay at $i$ \cite{Le2015coml}, and the corresponding received signal vector at $\mathcal{D}$ is given by   
\begin{equation}\label{eq:rx D}
    \vec{y}_{D_i}(\tau)=\frac{G_i}{\sqrt{\bar{l}_{|D_i|}}}y_{i}(\tau-\Delta){\vec{h}}_{D_i}+{\vec{n}}_{D}(\tau)
\end{equation}
where ${\vec{h}}_{D_i}$ denotes the fading channel coefficient vector of link $D_i$, and ${\vec{n}}_{D}(\tau)$ is the AWGN vector at $\mathcal{D}$ with variance $\sigma_D^2$. 

We define the interference-to-noise ratio (INR) at $i \in \mathcal{U}$ as $\Gamma_I\triangleq\sigma_I^2/\sigma_u^2$. The signal-to-interference-plus-noise ratios (SINRs) at $i$ and $\mathcal{D}$ can then be given, respectively, as 
\begin{equation}\label{e_sinrR}
\hspace{-1mm}\tilde\Gamma_{i}(a) = \frac{ a\Gamma_{S_i} }{(1-a)\Gamma_{S_i} +1}
\text{ and }
\tilde\Gamma_{\rm D_i}(a) = \frac{a\Gamma_{S_i}\Gamma_{D_i}}{\Gamma_{S_i}+\Gamma_{D_i}+1},
\end{equation}
where $\Gamma_{S_i}$ and $\Gamma_{D_i}$  are the SNRs of $S_i$ and $D_i$ links, respectively, given by
\begin{equation}
\Gamma_{S_i}\triangleq \frac{P_s |h_{S_i}|^2 }{\bar{l}_{|S_i|}\sigma_u^2(1+\Gamma_I)}
\text{ and } 
\Gamma_{D_i}\triangleq \frac{P_u}{\bar{l}_{|D_i|}\sigma_d^2}\sum_{\ell=1}^{N}|h_{\ell}|^2,
\label{Gamma-def}
\end{equation}
where $h_{\ell}$ is the $\ell$-th element of $\vec{h}_{D_i}$.

Following the Nakagami-$m$ fading model, $\sigma_{S_i}^2$ and $\sigma_{D_i}^2$ are the respective scale parameters of $S_i$ and $D_i$ channels with $m$ as the shape parameter. When $m$ is an integer, the fading power is Gamma distributed, and hence, RV $\Gamma_{S_i}$ follows a Gamma distribution with shape parameter $m$ and scale parameter $\bar\gamma_{S_i} \triangleq P_s \sigma_{S_i}^2/\paren{\sigma_u^2\bar{l}_{|S_i|}(\Gamma_I+1)} $. Similarly, RV $\Gamma_{D_i}$ also follows a Gamma distribution with shape parameter $Nm$ and scale parameter $\bar\gamma_{D_i} \triangleq P_u \sigma_{D_i}^2/\paren{\sigma_d^2\bar{l}_{|D_i|}}$. 
The symbols $f_X$, and $F_X$ are used to denote the probability density function (PDF) and the CDF of an RV $X$, respectively.

 
 The instantaneous achievable secrecy rate using $i\in \mathcal{U}$ as a relay can then be expressed as 
\begin{equation}\label{e_isr}
C_{i}(a) = \left[\ln(1+\tilde\Gamma_{\rm D_i}(a)) - \ln(1+\tilde\Gamma_{\rm i}(a)) \right]^+,
\end{equation}
where $[x]^+\triangleq\max\{0,x\}$ \cite{Atapattu2019coml}.

\section{Secrecy Performance Analysis With UAV Selection}\label{secrecy_perform}

Single relay selection is a widely used technique to improve the spectral efficiency, and the end-to-end (e2e) latency of relay networks, while achieving full diversity order \cite{Atapattu13_jsac}. Furthermore, relay selection reduces the synchronization issues compared to multiple-relay communications. Therefore, we employ relay selection, such that a single UAV is selected as the relay to assist communication between $\mathcal{S}$ and $\mathcal{D}$. Here, we select the UAV that maximizes the instantaneous secrecy rate of the system. 
The power allocation ratio $a$ is a key factor determining the secrecy rate of the system.
Therefore, firstly, we use the results in \cite{Atapattu2019coml} to deduce the optimal power allocation ratio for our setup.
To this end, the optimal power allocation ratio $a^*$ that maximizes the system secrecy rate is given by \begin{equation*}
a^*=\argmax_{0\leq a \leq 1}\,\,C_{i}(a) 
 = \argmax_{0\leq a \leq 1}\,\,\max\{\Psi(a),1\},
\end{equation*}
where $\Psi(a)\triangleq(1+\tilde{\Gamma}_{\rm D_i}(a))/(1+\tilde{\Gamma}_{\rm i}(a))$, and the second equality is written using the monotonicity of the logarithm function. Following \cite{Atapattu2019coml}, when $\Gamma_{D_i} < 1+1/\Gamma_{S_i}$, $\ie$, when $S_i$ or $D_i$ channels are encountering deep fading, $C_{i}(a)=0$ for any $a\in[0,1]$. This implies that the transmission is futile and $\mathcal{S}$ should be kept idle to save energy.
When this condition does not hold, the optimal power allocation ratio is given by
\begin{equation}
\label{e_rs}
a^*= \frac{1}{2}\left(1-\frac{1+\Gamma_{S_i}}{\Gamma_{S_i}\Gamma_{D_i}}\right),
\end{equation} 
which leads to  
\begin{equation}\label{e_opta}
\begin{split}
\tilde\Gamma_{\rm i}^* = \frac{\Gamma_{S_i} (\Gamma_{D_i}-1)-1}{\left(\Gamma_{S_i}+2\right) \Gamma_{D_i}+\Gamma_{S_i}+1},\,\, 
\tilde\Gamma_{\rm D_i}^* = \frac{\Gamma_{S_i} (\Gamma_{D_i}-1)-1}{2 \left(\Gamma_{D_i}+\Gamma_{S_i}+1\right)}, 
\end{split}
\end{equation}
where $\tilde\Gamma_{\rm i}^* =\tilde{\Gamma}_i(a^*)$ and $\tilde\Gamma_{\rm D_i}^* =\tilde{\Gamma}_{D_{i}}(a^*)$. 
It is not hard to show that $\tilde\Gamma_{\rm i}^*<1$, which implies that the achievable secrecy rate at the UAV is low. On the other hand, $\tilde\Gamma_{\rm D_i}^*$ is unbounded implying that with preferable channel realizations, $\mathcal{D}$ can achieve high secrecy rates. This also means $\tilde\Gamma_{\rm D_i}^*-\tilde\Gamma_{\rm i}^* \geq 0$ is possible, which implies the possibility of secure information transmission.

Once the optimal power allocation ratios for the links through all UAVs are determined, the UAV which maximizes the e2e secrecy rate, $i^*$, can be found as
\begin{equation}
 i^*=\argmax_{1\leq i \leq R}\,\,C_{i}(a^*).   
\end{equation}

Next, we analyze the achievable performance with the selected UAV. For this, we evaluate the SOP and the ASR of the system.
\subsection{SOP Analysis}
\label{sop_analysis}
A secrecy outage event (hereafter referred to as outage) occurs when $C_{i}^*(a^*)=0$, where $C_{i}^*(a^*)$ is the secrecy rate achievable with the best ($i^*$) UAV \cite{Atapattu2019coml}. It is important to note that without power allocation between $x_s$ and $x_j$, $\ie$, when $a=1$ or $a=0$, the system will always be in outage. For an example, when $a=1$, one can observe that $\tilde\Gamma_{i}> \tilde\Gamma_{D_i}$, since $\Gamma_{S_i}$ and $\Gamma_{D_i}$ are always greater than or equal to zero. This makes $\Psi(1)<1$, and $C_{i}(1)=0$. Therefore, it is essential to employ SBJ with a proper power allocation scheme. 

Using SBJ with the optimal power allocation, the SOP of the proposed UAV selection scheme is given by
\begin{align}
      \nonumber   SOP^*&=\Pr\brparen{C_{i}^*(a^*)\leq 0}=\Pr\brparen{\max_i{(C_{i}(a^*))} \leq 0}\\
    \nonumber    &\stackrel{(a)}{=} \prod_{i=1}^R \Pr\brparen{C_{i}(a^*) \leq 0} \\
        &\stackrel{(b)}{\approx} \sqparen{\Pr\brparen{C_{i}(a^*) \leq 0}}^{R}= \sqparen{SOP_i^*}^{R}.
        \label{joint}
\end{align}
where $(a)$ follows from the independence of the fading channel coefficients of all links, and in $(b)$, we have assumed that the $R$ UAVs are located in a manner that the average path loss values from $\mathcal{S}$ and $\mathcal{D}$ to each of the UAVs are approximately equal. $SOP_i^*$ is the SOP of the link via the $i$-th UAV.   
Using SBJ with the optimal power allocation scheme, we have already established that the source device will be idle if $\Gamma_{D_i} < 1+1/\Gamma_{S_i}$. It is not hard to see that the probability of this event captures the SOP of a link via the $i$-th UAV. An analytical expression for this probability is presented through the following lemma.
 \begin{lemma}
\label{sop_i}
The SOP of the link via the $i$-{th} UAV is given by
\begin{equation} \label{sop_i_eq}
\begin{split}
    & SOP_i^*= 1- \hspace{-1mm} \frac{2 \e{-\frac{1}{\bar\gamma_{D_i}}}}{\Gamma(m)} \hspace{1mm}
   \\ & \times 
    \sum_{k=0}^{Nm-1}\sum_{\ell=0}^{k} \frac{ \paren{\bar{\gamma}_{S_i}}^{\frac{\ell-k-m}{2}}K_{m-k+\ell}\paren{\frac{2}{\sqrt{\bar\gamma_{S_i}\bar\gamma_{D_i}}}}}{\paren{\bar{\gamma}_{D_i}}^{\frac{\ell+k+m}{2}} \paren{k-\ell}! \hspace{1mm} \ell!}
\end{split}
\end{equation}
where $K_{n}\left(\cdot\right)$ is the modified Bessel function of the second kind.
\end{lemma}
\begin{proof}
Since $ SOP_i^*=\Pr [ \Gamma_{D_i} < 1+\textstyle \frac{1}{\Gamma_{S_i}}]$, we have
\begin{align}
\nonumber & SOP_i^*  
= \hspace{-1mm}\int_{0}^{\infty} \hspace{-1mm} F_{\Gamma_{D_i}}\paren{ 1+\frac{1}{t}} f_{\Gamma_{S_i}}(t) dt \\\label{e_optsop}
\nonumber & \stackrel{(a)}{=} 1-  \hspace{-1mm}\sum_{k=0}^{Nm-1}\frac{\e{-\frac{1}{\bar\gamma_{D_i}}} }{k! \bar\gamma_{D_i}^{k} } \int_{0}^{\infty}\hspace{-1mm} \paren{ 1+\frac{1}{t}}^k \hspace{-1mm} \frac{t^{m-1}}{\Gamma(m)\bar\gamma_{S_i}^m} \hspace{1mm} \e{-\frac{1}{\bar\gamma_{D_i} t} -\frac{t}{\bar\gamma_{S_i}}} dt
\end{align}
where (a) follows from the substitution of the respective PDF and CDF of $\Gamma_{S_i}$ and $\Gamma_{D_i}$.
Applying the identity $\int_{0}^{\infty} x^{v-1}e^{-a/x-bx} dx=2 \left(a/b\right)^{\frac{v}{2}} K_{v}(2 \sqrt{a b})$ together with the binomial expansion 
completes the proof.
\end{proof}
Substituting (\ref{sop_i_eq}) in (\ref{joint}), gives us the SOP of the system.

\subsection{ASR Analysis}
\label{asr_analysis}
The ASR of a system is another important measure of the secrecy performance. The ASR is defined as the probabilistic average of the instantaneous secrecy rate of the system.  For the system model under consideration, with optimal power allocated SBJ and secrecy rate maximizing relay selection, the ASR can be expressed as
\begin{equation}{\label{Csro}}
        \bar{C} =\ES{C_{i}^*(a^*)}=\ES{\max_{1\leq i\leq R}\sqparen{\log\paren{\frac{1+\tilde\Gamma_{D_i}^*}{1+\tilde\Gamma_{i}^*}}}}. 
\end{equation}
To simplify (\ref{Csro}), from (\ref{e_opta}) one can observe that 
the ratio $\frac{\tilde\Gamma_{D_i}^*}{\tilde\Gamma_{i}^*}=\tilde\Gamma_{D_i}^*+1$. Substituting this relationship in (\ref{Csro}), gives us
\begin{equation*}
\bar{C}=\ES{\log\paren{1+\max_{1\leq i\leq R}\sqparen{\frac{\paren{\tilde\Gamma_{D_i}^*}^{2}}{2\tilde\Gamma_{D_i}^*+1}}}}.
\end{equation*}
To evaluate the expectation, let $\max\paren{z_i}=Y$ where
$z_{i}=\frac{\paren{\tilde\Gamma_{D_i}^*}^{2}}{2\tilde\Gamma_{D_i}^*+1}$. Thus, from integration by parts, 
\begin{equation}
\label{avg_asr}
 \bar{C}=\int_0^\infty \frac{1-F_Y(y)}{1+y} dy,   
\end{equation}
where $F_{Y}(y)=\sqparen{F_{z_i}(y)}^{R}$ using the i.i.d. assumption made earlier.
Furthermore, 
\begin{equation*}
\label{cdf_si}
    \begin{split}
      F_{z_i}(y) &=\Pr\brparen{ \frac{\paren{\tilde\Gamma_{D_i}^*}^{2}}{2\tilde\Gamma_{D_i}^*+1} \leq y}
      =\Pr\brparen{\paren{\tilde\Gamma_{D_i}^*-y}^{2} \leq y^{2}+y}\\
        &=F_{\tilde\Gamma_{D_i}^*}\paren{y+\sqrt{y^{2}+y}}.
    \end{split}
\end{equation*}
Thus, from (\ref{avg_asr}), we get
\begin{equation}
\label{asr_simple}
\bar{C} =  \int_{0}^{\infty} \frac{1-\sqparen{F_{\tilde\Gamma_{D_i}^*}\paren{y+\sqrt{y^{2}+y}}}^{R}}{\paren{1+y}}dy.    
\end{equation}
We obtain an expression for the CDF $F_{\tilde\Gamma_{D_i}^*}$ through the following lemma.
\begin{lemma}
\label{lemma2}
The CDF of RV $\tilde\Gamma_{D_i}^*$ is given by
\begin{align}
\label{e_cdfGD1}
\nonumber & F_{\tilde\Gamma_{\rm D_i}^*}(t) 
=1-\frac{2\e{ - 2t\paren{\frac{1}{\bar\gamma_{S_i}}+\frac{1}{\bar\gamma_{D_i}}}}}{\e{\frac{1}{\bar\gamma_{D_i}} } \Gamma(m)} 
\sum _{k=0}^{Nm-1} \sum _{\ell=0}^k \sum _{r=0}^{m-1}(2t)^{m-1-r}  \\ &\times(2t+1)^{k+r+1}\frac{ \binom{k}{\ell} \binom{m-1}{r} K_{r-\ell+1}\left(\frac{2(2 t+1)}{\sqrt{\bar\gamma_{S_i}\bar\gamma_{D_i}} }\right)}{k!\paren{\bar{\gamma}_{D_i}}^{k+\frac{r-\ell+1}{2}} \paren{\bar{\gamma}_{S_i}}^{m-\frac{r-\ell+1}{2}}}.
\end{align}
\end{lemma}
\begin{proof}
For $t\ge 0$, we can write
{\small
\begin{align}
\nonumber &F_{\tilde{\Gamma}_{\rm D_i}^*}(t) 
=\Pr\sqparen{\tilde{\Gamma}_{D_i}^*=0}+\Pr\sqparen{0<\tilde{\Gamma}_{D_i}^*\le t}\cr
&=\Pr\sqparen{\Gamma_{D_i}\leq 1+\frac{1}{\Gamma_{S_i}}} \cr 
& \quad +\Pr\sqparen{\Gamma_{D_i} > \frac{(2t+1)(1+\Gamma_{S_i})}{(\Gamma_{S_i}-2t)}, \Gamma_{S_i}<2t , \Gamma_{D_i}> 1+\frac{1}{\Gamma_{S_i}}} \\
\nonumber &\quad+ \Pr\sqparen{\Gamma_{D_i} \leq \frac{(2t+1)(1+\Gamma_{S_i})}{(\Gamma_{S_i}-2t)}, \Gamma_{S_i}> 2t , \Gamma_{D_i}> 1+\frac{1}{\Gamma_{S_i}}}\\\label{e_cdfGD11}  \nonumber &=F_{\Gamma_{S_i}}(2t)+\int_{2t}^{\infty} F_{\Gamma_{D_i}}\paren{\frac{\paren{2t+1}\paren{1+x}}{\paren{x-2t}}} f_{\Gamma_{S_i}}(x)dx.
\end{align}
}
Substituting the respective PDF and CDF of $\Gamma_{S_i}$ and $\Gamma_{D_i}$ 
and solving the integral applying the transformation $y=x-2t$ and the binomial expansion, yields (\ref{e_cdfGD1}). 
\end{proof} Substituting (\ref{e_cdfGD1}) in (\ref{asr_simple}) gives us an expression for ASR, which can be evaluated using numerical integration techniques.

\subsection{High SNR SOP Analysis}
\label{high_snr}
The derived SOP and ASR expressions are useful in obtaining exact performance values of the system. However, for certain applications, it is useful to obtain direct insights on the system behaviour in the high SNR regime. Hence, we derive approximations for the SOP in the high SNR regime and evaluate the diversity order of the system.
For the high SNR regime, we let $\bar\gamma=P_{s}/\sigma_u^2=\delta P_{u}/\sigma_d^2 \gg1$.

When $N=1$ and $m=1$, using the relation $x K_{1}(x)\approx 1$ for small $x$, (\ref{joint}) can be approximated as $SOP^*\approx (1-e^{-\frac{1}{\bar\gamma_{D_i}}})^{R}$,  giving a diversity order of $R$ for the system.
For an arbitrary integer $m$, when $N \geq 2$, $SOP_i^*$ can be lower bounded as
\begin{equation}
\label{asyout}
  \begin{split}
      SOP_i^*&\geq\max\brparen{\Pr\paren{\Gamma_{D_i}<1},\Pr\paren{\Gamma_{D_i}<\frac{1}{\Gamma_{S_i}}}} \\
&=\max \left\{\frac{\bar\gamma^{-Nm}(\delta \bar{l}_{|D_i|})^{Nm}}{(Nm)!},\right.\\
&\quad \quad \left.\paren{\frac{\delta \bar{l}_{|S_i|} \bar{l}_{|D_i|} (1+\Gamma_I)}{\bar\gamma^{2}}}^{m} \frac{\Gamma((N-1)m)}{m!(Nm-1)!}\right\},
  \end{split}  
\end{equation}
where the Taylor series expansions of $F_{\Gamma_{D_i}}(t)$ and $f_{\Gamma_{D_i}}(t)$ are used to compute the probabilities.
Furthermore, $SOP_i^*$ can be upper bounded as
\begin{equation}
\label{asyout2}
    \begin{split}
       SOP_i^*&\le\Pr\paren{\Gamma_{D_i}<2}+\Pr\paren{\Gamma_{D_i}<\frac{2}{\Gamma_{S_i}}}\\
        &=\paren{\frac{2\delta \bar{l}_{|S_i|} \bar{l}_{|D_i|} (1+\Gamma_I)}{\bar\gamma^{2}}}^{m} \frac{\Gamma((N-1)m)}{m!(Nm-1)!}\\&\quad \quad \quad \quad \quad \quad \quad \quad 
        +\frac{\bar\gamma^{-Nm}(2\delta \bar{l}_{|D_i|})^{Nm}}{(Nm)!}.
    \end{split}
\end{equation}

From \eqref{asyout} and \eqref{asyout2}, one can deduce that the diversity order of $2m$ can be achieved with a single UAV. 
Therefore, by selecting the best UAV out of $\mathcal{U}$ results in a diversity order of $2Rm$, when $N\geq 2$. It is interesting to note that the diversity order is independent of $N$.
\section{Numerical Results and Discussion}\label{s_num}

In this section, we present numerical results for validation of the theoretical results, and to draw insights to provide recommendations on UAV network design for secure communications. 
Radio propagation parameters ($\omega, \beta, \eta_{L} ({\small \rm dB}), \eta_{N} ({\small \rm dB})$) for the urban environmental conditions are given as (9.61, 0.16, 1, 20) \cite{UAV_environmentalParameters} and $P_s =P_u$.
It is assumed that $\Gamma_I=1$ in the numerical results.

\begin{figure}
    \centering{\includegraphics[scale=0.4]{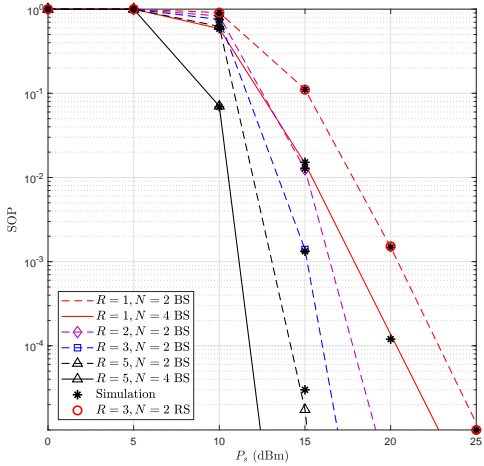}} \caption{The SOP vs. $P_{s}$ in an urban environment for SBJ with $m=2$.} \vspace{-0.2cm}  \label{Fig:sopASN}
    \end{figure}
\begin{figure}
    \centering{\includegraphics[scale=0.4]{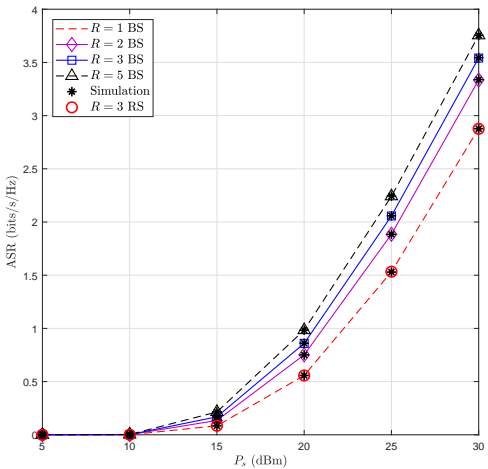}} \caption{The ASR vs. $P_{s}$ in an urban environment for $ N=2,$ and $m=2$.} \vspace{-0.2cm}  \label{Fig:asrAS}
\end{figure}

    

Firstly, Figs. \ref{Fig:sopASN} and \ref{Fig:asrAS} present the SOP and ASR performance vs. $P_s$ for different values of $R$ and $N$. In the legends, BS refers to the best UAV selection scheme, while RS refers to random UAV selection. One can clearly observe the excellent agreement of analytical results with simulation results, confirming their accuracy. Furthermore, the slopes of the SOP curves in the high SNR regime remains the same for $N=2$ and $N=4$, when $R$ is fixed. This confirms our observation of the independence between the diversity order and $N$ when $N >2$. Also, the best UAV selection outperforms random UAV selection, where one can observe that the SOP and the ASR of random UAV selection are similar to the case with a single UAV. This highlights the importance of UAV selection in achieving better secrecy performance.

\begin{figure}
    \centering{\includegraphics[scale=0.4]{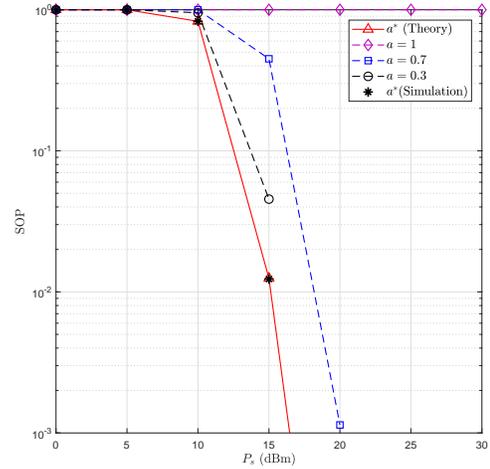}} \caption{SOP vs. $P_{s}$ in an urban environment for SBJ with different power allocation schemes when $N=2,$ and $m=2$.} \vspace{-0.2cm}  \label{Fig:sopPo}
    \end{figure}

Next, Fig. \ref{Fig:sopPo} illustrates the SOP performance of SBJ schemes for different power allocation ratios. As predicted through the analytical results, the system is always in outage when power is not allocated to the jamming signal ($a=1$). Furthermore, the performance gain achieved with the optimal power allocation compared to a heuristic power allocation scheme can be clearly identified.

\begin{figure}
\centering{\includegraphics[scale=0.4]{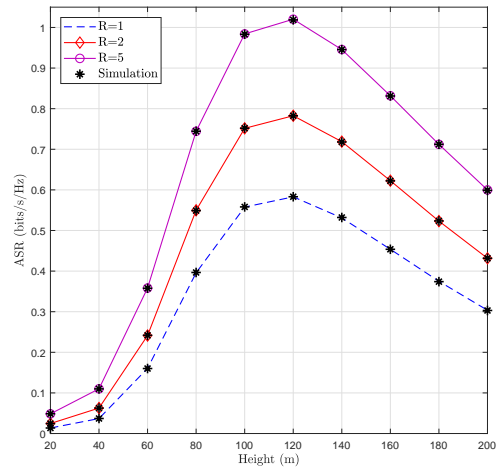}} \caption{ASR vs. $H_i$ for $R$ UAVs,  when $P_{s}=20$dBm, $N=2,$ and $m=2$, in an urban environment. } \vspace{-0.2cm}  \label{Fig:asrEH}
\end{figure}

Furthermore, Fig. \ref{Fig:asrEH} presents the ASR behaviour with the UAV altitude for different $R$ values. It is assumed that the set of relays are placed close to each other such that the link lengths remain almost the same for all UAVs. At first, ASR improves with the UAV altitude since the probability of a LoS link increases with $H_i$. However, one can identify the existence of an optimal altitude which maximizes the ASR, after which the ASR decreases due to the dominant effect of increased path loss. Therefore, it is important to choose a proper UAV altitude to achieve a favourable ASR for the system.

\begin{figure}
\centering{\includegraphics[scale=0.4]{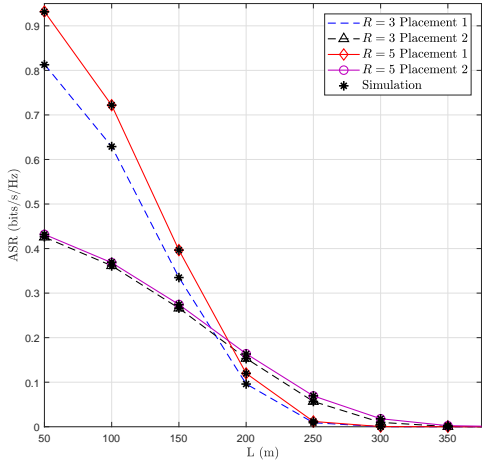}} \caption{ASR vs $D$ for different placement schemes, when $P_{s}=15$ dBm, $N=2,$ and $m=2$, in an urban environment.} \vspace{-0.2cm}  \label{Fig:asrRD_DR}
\end{figure}

Finally, Fig. \ref{Fig:asrRD_DR} shows the behaviour of the ASR with $L$, the distance between $\mathcal{S}$ and $\mathcal{D}$, for two UAV placement schemes. In Placement 1, UAVs are placed at the same altitude such that the projections of the UAVs on the ground plane equally divide the distance between $\mathcal{S}$ and $\mathcal{D}$. Placement 2 refers to the scenario where UAVs are placed at different altitudes such that their projections on the ground plane bisect the line between $\mathcal{S}$ and $\mathcal{D}$. The results show that Placement 1 is preferred for small values of $L$, while Placement 2 is preferred for larger values of $L$. When $L$ is small, there is a higher chance of finding an LoS link with a relay, and hence, we prefer the relay to be as close as possible to $\mathcal{S}$ and $\mathcal{D}$ to minimize the effect of path loss. This leads to Placement 1 being better for smaller values of $L$. On the other hand, when $L$ increases, increasing the altitude of the relay increases the possibility of connection establishment through an LoS link. This leads to Placement 2 being better for larger values of $L$.

\section{Conclusion}\label{s_con}
The secrecy outage probability and the average secrecy rate of UAV relay assisted communication were investigated. Source based jamming with optimal power allocation among the useful signal and the jamming signal was used to achieve secure communications through untrusted UAV relays. The best UAV selection scheme was proposed to improve the secrecy performance of the network. The results demonstrate the achievable secrecy performance gain with optimal power allocation and relay selection schemes, providing a diversity order of $2Rm$ for the system. Furthermore, useful network design insights such as locations and altitudes to place UAV relays were inferred from the numerical results.


\begin{thebibliography}{10}

\bibitem{Asadpour2014magcom}
M.~{Asadpour}, B.~{Van den Bergh}, D.~{Giustiniano}, K.~A. {Hummel},
  S.~{Pollin}, and B.~{Plattner}, ``Micro aerial vehicle networks: an
  experimental analysis of challenges and opportunities,'' {\em {IEEE} Commun.
  Mag.}, vol.~52, pp.~141--149, July 2014.

\bibitem{zhao2018coml}
W.~{Zhao}, G.~{Wang}, S.~{Atapattu}, C.~{Tellambura}, and H.~{Guan}, ``Outage
  analysis of ambient backscatter communication systems,'' {\em {IEEE} Commun.
  Lett.}, vol.~22, pp.~1736--1739, Aug. 2018.

\bibitem{atapattu2017coml}
S.~Atapattu, P.~Dharmawansa, C.~Tellambura, and J.~Evans, ``Exact outage
  analysis of multiple-user downlink with {MIMO} matched-filter precoding,''
  {\em {IEEE} Commun. Lett.}, vol.~21, pp.~2754--2757, Dec. 2017.

\bibitem{Zhao2019wcom}
N.~{Zhao}, W.~{Lu}, M.~{Sheng}, Y.~{Chen}, J.~{Tang}, F.~R. {Yu}, and
  K.~{Wong}, ``{UAV}-assisted emergency networks in disasters,'' {\em IEEE
  Wireless Commun.}, vol.~26, pp.~45--51, Feb. 2019.

\bibitem{MultipleUAVrelays}
Y.~{Chen}, N.~{Zhao}, Z.~{Ding}, and M.~{Alouini}, ``Multiple {UAV}s as relays:
  Multi-hop single link versus multiple dual-hop links,'' {\em {IEEE} Trans.
  Wireless Commun.}, vol.~17, pp.~6348--6359, Sept. 2018.

\bibitem{Wang2017coml}
Q.~{Wang}, Z.~{Chen}, W.~{Mei}, and J.~{Fang}, ``Improving physical layer
  security using {UAV}-enabled mobile relaying,'' {\em {IEEE} Commun. Lett.},
  vol.~6, pp.~310--313, June 2017.

\bibitem{Wang2018acc}
Q.~{Wang}, Z.~{Chen}, H.~{Li}, and S.~{Li}, ``Joint power and trajectory design
  for physical-layer secrecy in the {UAV}-aided mobile relaying system,'' {\em
  IEEE Access}, vol.~6, pp.~62849--62855, 2018.

\bibitem{Zhong2019coml}
C.~{Zhong}, J.~{Yao}, and J.~{Xu}, ``Secure {UAV} communication with
  cooperative jamming and trajectory control,'' {\em {IEEE} Commun. Lett.},
  vol.~23, pp.~286--289, Feb. 2019.

\bibitem{Liu2018jcn}
H.~{Liu}, S.~{Yoo}, and K.~S. {Kwak}, ``Opportunistic relaying for low-altitude
  {UAV} swarm secure communications with multiple eavesdroppers,'' {\em J.
  Commun. and Networks}, vol.~20, pp.~496--508, Oct. 2018.

\bibitem{Zhou2018tvt}
Y.~{Zhou}, P.~L. {Yeoh}, H.~{Chen}, Y.~{Li}, R.~{Schober}, L.~{Zhuo}, and
  B.~{Vucetic}, ``Improving physical layer security via a {UAV} friendly jammer
  for unknown eavesdropper location,'' {\em {IEEE} Trans. Veh. Technol.},
  vol.~67, pp.~11280--11284, Nov. 2018.

\bibitem{Li2019coml}
A.~{Li}, Q.~{Wu}, and R.~{Zhang}, ``{UAV}-enabled cooperative jamming for
  improving secrecy of ground wiretap channel,'' {\em {IEEE} Commun. Lett.},
  vol.~8, pp.~181--184, Feb. 2019.

\bibitem{Tang2019ifs}
J.~{Tang}, G.~{Chen}, and J.~P. {Coon}, ``Secrecy performance analysis of
  wireless communications in the presence of {UAV} jammer and randomly located
  {UAV} eavesdroppers,'' {\em IEEE Trans. Inform. Forensics and Security},
  2019.

\bibitem{Ye2019wcoml}
J.~{Ye}, C.~{Zhang}, H.~{Lei}, G.~{Pan}, and Z.~{Ding}, ``Secure {UAV}-to-{UAV}
  systems with spatially random {UAV}s,'' {\em IEEE Wireless Commun. Lett.},
  vol.~8, pp.~564--567, Apr. 2019.

\bibitem{Yener2010it}
X.~He and A.~Yener, ``Cooperation with an untrusted relay: A secrecy
  perspective,'' {\em {IEEE} Trans. Inform. Theory}, vol.~56, pp.~3807--3827,
  Aug. 2010.

\bibitem{Atapattu2019twc}
S.~{Atapattu}, N.~{Ross}, Y.~{Jing}, Y.~{He}, and J.~S. {Evans},
  ``Physical-layer security in full-duplex multi-hop multi-user wireless
  network with relay selection,'' {\em {IEEE} Trans. Wireless Commun.},
  vol.~18, pp.~1216--1232, Feb. 2019.

\bibitem{Wang2018jsac}
H.~{Wang}, J.~{Wang}, G.~{Ding}, J.~{Chen}, Y.~{Li}, and Z.~{Han}, ``Spectrum
  sharing planning for full-duplex {UAV} relaying systems with underlaid {D2D}
  communications,'' {\em {IEEE} J. Select. Areas Commun.}, vol.~36,
  pp.~1986--1999, Sept. 2018.

\bibitem{Atapattu2019coml}
S.~{Atapattu}, N.~{Ross}, Y.~{Jing}, and M.~{Premaratne}, ``Source-based
  jamming for physical-layer security on untrusted full-duplex relay,'' {\em
  {IEEE} Commun. Lett.}, vol.~23, pp.~842--846, May 2019.

\bibitem{Atapattu13_jsac}
S.~Atapattu, Y.~Jing, H.~Jiang, and C.~Tellambura, ``Relay selection and
  performance analysis in multiple-user networks,'' {\em {IEEE} J. Select.
  Areas Commun.}, vol.~31, pp.~1517--1529, Aug. 2013.

\bibitem{Xiong2016twc}
X.~{Xiong}, X.~{Wang}, T.~{Riihonen}, and X.~{You}, ``Channel estimation for
  full-duplex relay systems with large-scale antenna arrays,'' {\em {IEEE}
  Trans. Wireless Commun.}, vol.~15, pp.~6925--6938, Oct. 2016.

\bibitem{NakagamiM}
H.~{Liu}, S.~{Yoo}, and K.~S. {Kwak}, ``Opportunistic relaying for low-altitude
  {UAV} swarm secure communications with multiple eavesdroppers,'' {\em Journal
  of Communications and Networks}, vol.~20, pp.~496--508, Oct. 2018.

\bibitem{ChannelModelUAV}
A.~{Al-Hourani}, S.~{Kandeepan}, and S.~{Lardner}, ``Optimal {LAP} altitude for
  maximum coverage,'' {\em {IEEE} Wireless Commun. Lett.}, vol.~3,
  pp.~569--572, Dec. 2014.

\bibitem{Ngoc2014jsac}
L.~J. Rodríguez, N.~H. Tran, and T.~Le-Ngoc, ``Performance of full-duplex {AF}
  relaying in the presence of residual self-interference,'' {\em {IEEE} J.
  Select. Areas Commun.}, vol.~32, pp.~1752--1764, Sept. 2014.

\bibitem{Le2015coml}
T.~P. Do and T.~V.~T. Le, ``Power allocation and performance comparison of full
  duplex dual hop relaying protocols,'' {\em {IEEE} Commun. Lett.}, vol.~19,
  pp.~791--794, May 2015.

\bibitem{UAV_environmentalParameters}
R.~I. {Bor-Yaliniz}, A.~{El-Keyi}, and H.~{Yanikomeroglu}, ``Efficient 3-{D}
  placement of an aerial base station in next generation cellular networks,''
  in {\em Proc. {IEEE} Int. Conf. Commun. (ICC)}, May 2016.

\end{thebibliography}
\end{document}